\pdfoutput=1
\documentclass[conference]{IEEEtran}
\usepackage{amsthm}
\newtheorem{definition}{Definition}
\newtheorem{theorem}{Theorem}

\usepackage{cite}
\usepackage{amsmath,amssymb,amsfonts}
\usepackage{algorithmic}
\usepackage{graphicx}
\usepackage{textcomp}
\usepackage{xcolor}
\usepackage{listings}
\usepackage{booktabs}
\usepackage[hidelinks]{hyperref}
\def\BibTeX{{\rm B\kern-.05em{\sc i\kern-.025em b}\kern-.08em
    T\kern-.1667em\lower.7ex\hbox{E}\kern-.125emX}}
\begin{document}

\title{Ghost Vectors: Soft-Deleted Embeddings Remain Reconstructible in HNSW Vector Databases}

\author{\IEEEauthorblockN{Chandranil Chakraborttii, Jackeline García Alvarado, Sitora Abdulofizova, and Shivanshu Dwivedi}
\IEEEauthorblockA{Trinity College, USA}}

\maketitle

\begin{abstract}
Retrieval-augmented generation (RAG) allows large language models to access external and private corpora for factual, domain-specific responses.  Modern RAG pipelines use hierarchical navigable small world (HNSW) vector databases for efficient similarity search. When a user requests data deletion, the systems typically only mark the record as deleted, leaving the embedding on disk physically unchanged. This soft-delete operation raises compliance concerns under data-erasure and retention requirements such as GDPR Article 17 and HIPAA. Analysis on three HNSW implementations confirms that deleted vectors remain physically recoverable by accessing the raw index files at the storage layer, bypassing API access. Using the Vec2Text inversion model without domain-specific fine-tuning, we show this vulnerability on multiple real-world datasets and data modalities. On Wikipedia biographical living persons dataset~(BLP), we successfully recover 25.5\% of exact person names and 46.4\% of geographic locations (ROUGE-L $0.185 \pm 0.062$). Recovery reaches 100\% for both patient age and gender markers (ROUGE-L $0.290 \pm 0.033$) on highly structured, sensitive data~(NIH Synthea dataset). On soft-deleted image embeddings, we show 100\% tissue classification on histopathology patches ($p=1.02\times10^{-7}$) and top-1 identity recovery reaches 99\% on facial embeddings~($p<10^{-12}$). This work introduces Epoch Key Rotation, which encrypts vectors and discards the key upon deletion. Epoch key rotation reduces observed PII recovery to 0\% and completes in 2.5 ms for 500 deleted vectors (approximately 0.005 ms/record). Additionally, it generates an ECDSA-signed cryptographic proof~($\pi$) as an auditable record of the deletion event. 
\end{abstract}


\section{Introduction}
\label{sec:intro}
Retrieval-Augmented Generation (RAG) technologies are being rapidly deployed across highly regulated domains~(such as finance or healthcare) using HNSW-based vector databases in the backend to store and retrieve sensitive data. Regulations such as Article 17 of the General Data Protection Regulation (GDPR) allow users to request immediate and full removal of their data~\cite{gdpr2016,edpb2019erasure} from client systems. However, these implementations raise urgent regulatory concerns. Vector databases typically implement deletion as a `soft deletion' at the metadata level by design  to preserve performance. Under data minimization obligations~\cite{villaronga2018humans}, user data must not be retained beyond its intended purpose \cite{politou2018forgetting, villaronga2018humans} --- soft-delete is potentially inconsistent with these regulations indefinitely. Treating no response from an API as confirmation of deletion assumes that data has been removed, even though it may still physically exist. Query-side embedding transformations have been introduced to reduce the risk of inversion attacks; however, it  applies only to active queries and does not provide any protection once a vector enters secondary storage.

Our threat model operates at the secondary storage layer, where active query defenses are no longer effective. We assume an adversary has access to the raw storage data, bypassing all API restrictions. Since soft-delete only modifies API metadata, the vectors remain physically present in the HNSW index~(invisible to queries, but intact on disk). The adversary reads the HNSW binary index (e.g., ChromaDB’s index.bin file), extracts soft-deleted vectors using \texttt{get\_items()}, and feeds them into a pre-trained Vec2Text corrector~\cite{morris2023vec2text} without any fine-tuning. The recovered vectors are inverted back into plaintext, recovering sensitive information at higher rates than accepted GDPR and HIPAA thresholds.  We show the full attack pipeline in Figure~\ref{fig:pipeline} and describe it in detail in Section~\ref{sec:attack}. This work introduces \textbf{Epoch Key Rotation} as a storage-layer solution to this vulnerability. The approach encrypts vectors with AES-256-CTR and discards the key, reducing PII recovery to 0\% in our experiments. A signed cryptographic proof is also generated, providing an auditable record of the deletion event. This paper makes the following five contributions:

\noindent (1) We \textbf{empirically show} that soft-deleted HNSW vectors remain physically present, even after deletion requests are confirmed via API on three independent systems (ChromaDB, FAISS, and Weaviate).

\noindent (2) On the Wikipedia BLP dataset, 25.5\% of person names and 46.4\% of locations were recovered (n=98, ROUGE-L = 0.185 ± 0.062). On the NIH-standard Synthea dataset, recovery reached 100\% for age and gender markers under corrector-aligned clinical text. In the synthetic clinical dataset, patient surname recovery reached 49.9\% (CI: 46.8--53.0), indicating significant re-identification risk in our evaluation. On MIMIC-III, reconstruction remained low (ROUGE-L $<$ 0.010), suggesting that text reconstruction quality depends on corrector alignment. However, physical persistence is universal, representing a latent threat when stronger domain-adapted inversion models become available.

\noindent  (3) \textbf{The threat is indefinite}, since a backup snapshot can be used for attack with identical reconstruction quality (ROUGE-L=0.207). Vec2Text reconstructions transfer across three structurally different surrogate models (ROUGE-L$\geq$0.19, cosine similarity 0.86--0.90), confirming the attackers does not require knowledge of the victim's embedding model. The threat also shows low variance across the tested scales since ROUGE-L variance across {1K, 5K, 10K, 50K, 100K} remains at 0.016.

\noindent  (4) \textbf{The threat extends to image embeddings and structural metadata.} Soft-deleted facial images result in 99\% top-1 identity accuracy. Similarly, histopathology patches (PathMNIST, $N=1000$) are recovered at 100\% tissue-class accuracy ($p = 1.02 \times 10^{-7}$). Additionally, we show that ghost nodes remain structurally active in the HNSW graph, degrading live search results and creating a timing side-channel that reduces query latency by 4.3\% under true deletion (Table IX).

\noindent  (5) \textbf{Epoch key rotation reduces all PII recovery to 0\%} at a rate of 0.005ms/record and generates an ECDSA-signed cryptographic proof supporting auditable erasure workflows in regulated settings~(e.g, GDPR Art. 17 or HIPAA). For deployments requiring strict isolation, we later show that standard SQLite key deletion leaves key data recoverable from free pages, motivating hardware-isolated key storage (TPM/TrustZone).

\section{Background}
\label{sec:background}

\textbf{I. HNSW Vector Indexes and Soft-Deletion:}
In vector databases, HNSW graphs are the most common index structure for approximate nearest neighbor search~\cite{malkov2018hnsw,babenko2016efficient}. HNSW builds a multi-layer graph where each node represents a vector, edges connect it to its $M$ nearest neighbors, and layer-by-layer traversal takes $O(\log n)$ search time. To minimize latency, systems built on HNSW handle deletion with metadata flags, instead of erasing data from disk directly. For example, ChromaDB calls \texttt{delete(ids=[...])} to update an SQLite metadata table, with subsequent queries skipping the flagged entries~\cite{chromadb}. However, the raw \texttt{index.bin} file is not modified. FAISS \texttt{IndexHNSWFlat} provides no \texttt{delete()} method, so all vectors accumulate indefinitely~\cite{johnson2019faiss,douze2024faiss}. Weaviate~\cite{weaviate} and Pinecone~\cite{pinecone} follow the same pattern. Thus, soft-delete represents a deliberate design trade-off, not a bug in implementation.

\textbf{II. Text Embedding Inversion (Vec2Text):}
\label{sec:background_vec2text}
Inversion models such as Vec2Text~\cite{morris2023vec2text} have been developed to reconstruct embeddings back into semantically equivalent text. Starting from a candidate sequence produced by beam search, the corrector runs T iterative steps, each reducing the embedding distance between the candidate and the target vector. A pre-trained corrector (\texttt{gtr-base}) learns to associate token-level modifications using embedding residuals, with the weights staying frozen at inference. Since Vec2Text generates plain text than model-specific artifacts, it can invert embeddings across different architectures~(without knowing the victim's exact embedding model). The \texttt{gtr-t5-base}~\cite{ni2021large} retriever used in this work is a pretrained dense retrieval model~\cite{reimers2019sentence}. Prior work reports high reconstruction accuracy when training and evaluation are performed on closely matched text distributions. However, they have not been studied on soft-deleted records inside a live vector database.

\textbf{III. GDPR Article 17 and Verifiable Erasure: }
GDPR Art.~17(1) grants data subjects the right to erasure `without undue delay'~\cite{gdpr2016}. EDPB Guidelines 05/2019~\cite{edpb2019erasure} specify that erasure must be both \emph{verifiable} and \emph{irreversible}; suppressing records from query results alone does not satisfy this requirement. Art.~5(1)(e) additionally limits retaining data beyond its intended purpose~\cite{gdpr2016}. Current soft-delete designs may violate this indefinitely, regardless of stated policy. In these systems, enforcement is treated as a self-certification. However, from a systems perspective, functional deletion (i.e., an API returns no record) and storage-layer erasure are not equivalent. After a soft-delete command, a controller may treat the API response as sufficient evidence of deletion without checking the underlying state. Even with vectors physically present on the disk, the controller can still truthfully report that the API has no records for that ID. This paper identifies these two gaps and proposes a storage-layer solution to address them.

\begin{figure*}[t]
\centering
\scalebox{0.9}{
\includegraphics[width=\textwidth]{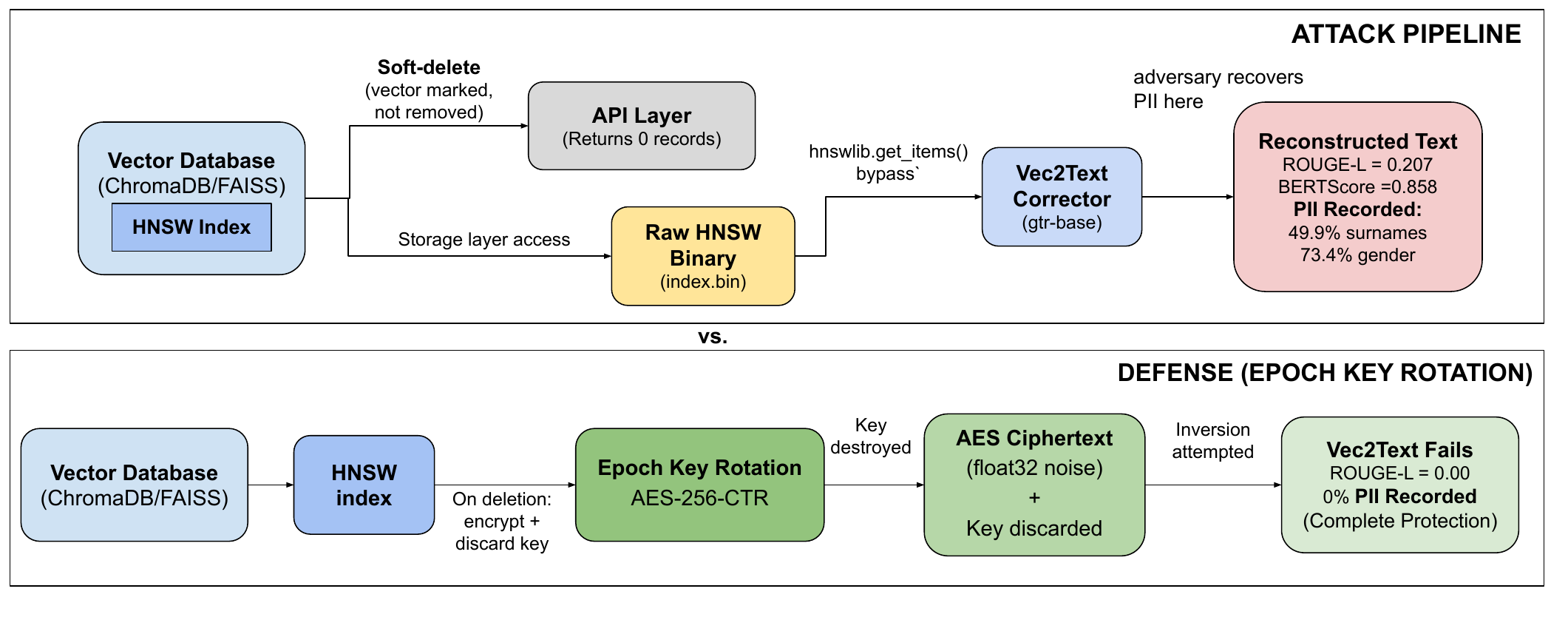}}
\vspace{-0.5cm}
\caption{Ghost vector attack pipeline (top) and epoch key rotation defense (bottom)}
\label{fig:pipeline}
\end{figure*}

\section{Threat Model}
\label{sec:threat_model}

\subsection{System and Adversary Model}
\textbf{Threat Model}: 
Three actors participate in our threat model: (i) a \textbf{Data Subject} that exercises the right to erasure under modern data complaince standards; (ii) a \textbf{Data Controller} operating a RAG system on top of an HNSW-based vector store (Weaviate, FAISS, or ChromaDB) who issues a soft-delete in response to deletion requests and treats the no API response as evidence of erasure; and (iii) a \textbf{Breach Adversary} with read access to raw storage who acts after the deletion has been recorded.  

The Data Controller embeds documents using the \texttt{gtr-t5-base} encoder (\(d=768\)) to build an HNSW index for RAG queries. After soft-delete, the API correctly returns no records for the removed IDs — but the controller performs no check at the storage layer, treating API silence as evidence of deletion.

\textbf{Adversary Model:} The breach adversary has \textbf{read access} to the raw storage directory (e.g, insider access, cloud storage mis-configuration, or backup data access). The adversary has \textbf{knowledge} of the embedding model that was used (often publicly available from system docs). Its \textbf{goal} is to reconstruct the semantic content of deleted data, i.e., records that the controller has declared erased. The plaintext, decryption keys, and API endpoints remain inaccessible.

Backup systems keep HNSW binary files regardless of API state, and entire database directories have been exposed in prior storage breaches~\cite{potluri2024deep}. A malicious insider (such as a database administrator or a cloud provider employee) can access raw files without causing API audit logs to be triggered.  Copying \texttt{index.bin} once is sufficient. The threat is not limited by deletion time, as soft-deleted vectors can remain unchanged in a backup for months after deletion.  The attacker takes advantage of the indefinite window that starts at deletion and ends only when the storage media is either physically destroyed or cryptographically overwritten. We empirically validate the indefinite threat window and scope of this threat model in Section~\ref{sec:persistence}.

\subsection{Formal Security Analysis}
\label{sec:proofs}
We formalize two main claims: (1) soft-delete leaves recoverable semantic information in storage, and (2) epoch key rotation removes that information under standard cryptographic assumptions. Both are expressed in terms of the Semantic Recovery Advantage~(SRA). We use $\mathcal{D}$ and $\mathcal{E}$ to represent a deletion process and an embedding function respectively. We define SRA of adversary $\mathcal{A}$ as:
\[
\mathsf{Adv}^{\mathsf{sem}}_{\mathcal{A}}(\lambda) = \Pr\left[\mathsf{EntityRecovery}\left(\mathcal{A}(\mathcal{D}(\mathcal{E}(x)))\right) \geq \theta\right]
\]
where $\theta = 0.2$ is the re-identification risk threshold, consistent with HIPAA Safe Harbor~\cite{hhs2012} to denote meaningful semantic recovery, and $x$ represents a plaintext record.

\begin{theorem}[Soft-Delete Risk]
\label{thm:softdelete_insecurity}
$\mathsf{Adv}^{\mathsf{sem}}_{\mathcal{A}}(\lambda)$ is empirically non-negligible under soft-delete ($\mathcal{D} = \mathsf{SoftMark}$), due to the persistence of the original embedding vector in raw storage (e.g., 0.499 surname recovery in Table~\ref{tab:pii_recovery}).
\end{theorem}

\begin{theorem}[Epoch Rotation Security]
\label{thm:epoch_security}
Assuming AES-256-CTR is used with securely managed and irrecoverably deleted keys, ($\mathcal{D} = \mathsf{EpochRotate}$),
$\mathsf{Adv}^{\mathsf{sem}}_{\mathcal{A}}(\lambda)$ is negligible in $\lambda$ for any PPT adversary $\mathcal{A}$.
\end{theorem}

\begin{proof}
Assume $\mathcal{A}$ achieves non-negligible semantic advantage over epoch key rotation. We construct a distinguisher $\mathcal{B}$ against AES-256-CTR as follows: \\
The AES-256-CTR challenger sends a challenge ciphertext $c^*$ to $\mathcal{B}$, which converts it into a float32 embedding vector and sends it to $\mathcal{A}$. $\mathcal{B}$ outputs 1 if entity recovery from $\mathcal{A}(c^*)$ exceeds $\theta$; otherwise 0. 
$\mathcal{A}$ should not recover meaningful semantic content with non-negligible probability when $c^*$ is truly random. By contradiction, suppose \(A\) could still recover meaningful semantic content from ciphertext produced under epoch rotation with non-negligible probability. Because AES-CTR is widely considered as IND-CPA secure under the standard PRP assumptions~\cite{bellare1997relations,nist2001aes,bellare2000authenticated}, this resultant gap gives $\mathcal{B}$ a non-negligible advantage in distinguishing AES-256-CTR output from random.  Therefore, under these assumptions, no PPT $\mathcal{A}$ achieves non-negligible semantic recovery advantage under epoch rotation.
\end{proof}

The bound in Theorem~\ref{thm:epoch_security} does not depend on Vec2Text or any specific corrector. Future inversion models cannot weaken it, because the ciphertext carries no semantic structure for them to recover. This argument depends on the security of the underlying AES-256 encryption scheme and on correct key disposal.
However, the argument fails if the system leaves the decryption key on disk. If the decryption key remains accessible on disk somewhere, every guarantee in Theorem~\ref{thm:epoch_security} no longer applies. Theorem~\ref{thm:epoch_security} requires erasing epoch keys permanently and irrevocably. SQLite's standard \texttt{DELETE} simply marks rows for overwrite without zeroing disk sectors, leaving key material recoverable from free pages. Hardware isolation (ARM TrustZone, HSM) or strict memory zeroing before deletion (\texttt{OPENSSL\_cleanse()}) in deployments is needed for full-proof secure key disposal.

\section{Ghost Vector Attack Setup}
\label{sec:attack}
\subsection{Threat Environment and Setup}
\label{sec:setup}
The attack is evaluated across three HNSW implementations and multiple text and multimodal datasets. Figure~\ref{fig:pipeline} shows the full~(attack and defense) pipeline, including storage-layer extraction, Vec2Text inversion, and epoch key rotation. Vector extraction runs entirely from Python scripts that read the on-disk index files directly, never touching the database query API or any defense layered on top of it. All experiments were conducted with PyTorch 2.0 on a dedicated node equipped with AMD Ryzen 7960X CPU with 96GB system RAM, and a NVIDIA RTX A6000 GPU with 48 GB VRAM.

\subsection{Evaluated Systems}
We evaluate ChromaDB (using \texttt{hnswlib}), FAISS (using \texttt{IndexHNSWFlat}), and Weaviate (in embedded mode).  For each system, we insert 5,000 embedded records, and soft-delete 10\% of each dataset using \texttt{default\_rng(43)}. Either \texttt{index.reconstruct\_n()} (for FAISS) or \texttt{hnswlib.get\_items()} (for ChromaDB) are used to extract vectors directly at the storage layer.

\subsection{Datasets}
\label{sec:setup_text}

\subsubsection{Text Datasets:}
(1) The baseline dataset used in Section 3 consists of the first 5,000 articles from wikimedia/wikipedia 20231101.en, trimmed to 512 characters and embedded using \texttt{gtr-t5-base} with mean pooling and L2 normalization. Reconstruction is evaluated across $N \in {1\text{K}, 5\text{K}, 10\text{K}, 50\text{K}, 100\text{K}}$, soft-deleting 10\% per scale.  \\
(2) Real-world privacy~(PII) leakage is measured on a Biographical Living Persons (BLP) subset of \(N=98\) records. Wikipedia baseline above is used to confirm scale. \\
(3) Since real-world healthcare data is sensitive, we instead use 1,000 synthetic clinical records generated with Faker~\cite{hameed2025generating} — each with complaints, diagnoses, and drug bundles. 58.8\% of records received one of five clinical abbreviation transformations (e.g., `Acute myocardial infarction' $\rightarrow$ `AMI', `twice daily' $\rightarrow$ `BID') before the embedding step, simulating real-world distortions. We also include two clinical stress tests: 1,000 Synthea EHR records with multi-year patient histories and 100 MIMIC-III Demo~\cite{johnson2016mimic} intensive-care notes containing physician shorthand. For the cryptographic-erasure baseline, we encrypt raw float32 bytes with AES-256-CTR and immediately discard the key.

\subsubsection{Multimodal Datasets:}
\label{sec:setup_multimodal}
Two image benchmarks complement the text experiments. Biometric identity recovery is tested on Labeled Faces in the Wild (LFW, $N = 4{,}324$ across 158 identities)~\cite{huang2007lfw}. Medical tissue-class recovery is evaluated on PathMNIST with \(N=1{,}000\) images~\cite{yang2023medmnist}. Both are embedded with CLIP ViT-H/14 ($d = 1{,}024$). A $k$-NN classifier ($k = 5$, cosine metric) handles downstream evaluation. LFW results are reported as mean ± std across 5 random seeds {$\{42, 43, 44, 99, 100\}$.

\subsection{Inversion Model \& Evaluation Metrics}
For the Vec2Text inversion attack (described earlier), we use \texttt{load\_pretrained\_corrector(`gtr-base')} and set \texttt{beam\_width=4}, \texttt{n\_steps=20} (where 10 steps reach a plateau). Clinical entity recovery requires an exact string match after lemmatization. For verifying real data extraction instead of simple keyword matching, ROUGE-L~\cite{lin2004rouge} is used that tracks the longest common subsequence between the reconstruction and the ground-truth text, penalizing structural reformulation highly. We also use 
BERTScore F1~\cite{zhang2020bertscore} to measure semantic preservation~(catches the cases ROUGE-L misses, i.e, where wording changes but meaning holds).

\section{Empirical Findings}
\subsection{Storage-Layer Persistence}
\label{sec:persistence}
Soft-delete works at the metadata layer instead of the storage layer. We use $N=5{,}000$ Wikipedia records from three databases to initialize HNSW instances using the standard API of each system. Then, as mentioned earlier, we soft-delete 10\% of each dataset using  a fixed random seed~\texttt{default\_rng(43)}. Either \texttt{index.reconstruct\_n()} (for FAISS) or \texttt{hnswlib.get\_items()} (for ChromaDB) were used to extract vectors. For Vec2Text, we use \texttt{load\_pretrained\_corrector(`gtr-base')} with \texttt{beam\_width=4}, \texttt{n\_steps=20}.  AES-256-CTR is applied to raw float32 bytes with instant key discard in the cryptographic-erasure baseline.  The results confirm that removed vectors still remain extractable~(Table~\ref{tab:multidb}).

\textbf{ChromaDB:} 
A ChromaDB instance with $N=5{,}000$ Wikipedia embedded records was created using gtr-t5-base ($d=768$). After soft-deleting 500 records using the standard \texttt{delete()} API, the directory size remained unchanged. Soft-deleted vectors take up identical disk space as live vectors. Deleted records are correctly hidden by the API layer, with \texttt{col.get(ids=[`42'])} returning no results. The hnswlib API, on the other hand, exposes \texttt{get\_items(list\_ids)}, which returns vectors for each stored ID, independent of deletion state~(for the storage layer, shown above). All 500 soft-deleted vectors are extracted using the \texttt{get\_items()} method, and cosine similarity $> 0.999$ is verified using the original embeddings. This confirms in ChromaDB, soft-delete is implemented as a metadata operation, not as a physical storage operation.

\textbf{FAISS and Weaviate:} 
FAISS does not support deletion. Since FAISS IndexHNSWFlat~\cite{johnson2019faiss} provides no \texttt{delete()} or \texttt{remove\_ids()} methods, \texttt{RuntimeError: not implemented} is raised when \texttt{index.remove\_ids()} is called. Since the HNSW graph traversal algorithm depends on all nodes remaining connected, all vectors added to IndexHNSWFlat are kept indefinitely. Using \texttt{index.reconstruct\_n(0, N)}, we recover all $N=5{,}000$ vectors, including the 500 `soft-deleted' subset. Vec2Text inversion on FAISS-extracted vectors  results in ROUGE-L = 0.167, slightly below ChromaDB (0.207) but still indicative of non-trivial text recovery in our evaluation.  The difference is expected: \texttt{reconstruct\_n} pulls raw floats directly from the flat index, bypassing the proximity-graph structure that ChromaDB's hnswlib backend maintains. The vectors come out slightly noisier with lower reconstruction quality.  Weaviate 1.26.6 behaves differently at the storage layer — disk size unexpectedly \textit{grows} after logical deletion ($52.5\text{MB} \rightarrow 55.4\text{MB}$), rather than contracting. HNSW slots are not reclaimed; deletion appends metadata instead. The API correctly returns no results for removed IDs. Weaviate's proprietary Go-based storage format blocks direct hnswlib extraction, but the observed disk growth is consistent with persistence rather than physical removal, aligning with FAISS and ChromaDB results.

\subsection{Reconstruction Results (Text)}
\label{sec:results}
\subsubsection{Baseline Extraction (BLP and Wikipedia)} 
Soft-deleted Wikipedia records retain substantial recoverable semantic content. We tested 500 soft-deleted records and recovered original text with a ROUGE-L score of 0.207 ($\pm 0.07$) and BERTScore F1 of 0.858 (Table~\ref{tab:reconstruction}). 97.8\% of reconstructed records scored above 0.1 on ROUGE-L. Hard deletion eliminated recoverable semantic content. Under the AES encryption and key-discard baseline, reconstruction produced ROUGE-L = 0.000 for all 500 records, showing that the original semantic information was no longer recoverable. To contextualize the semantic similarity scores, we also used a random baseline computed from 500 unrelated Wikipedia phrase pairs, which produced a BERTScore of 0.606. This value is used as a random-reference baseline for semantic similarity. In contrast, soft-deleted records produced higher scores, with BERTScore values concentrated around 0.85, as shown in Figure~\ref{fig:bertscore}. A Wilcoxon signed-rank test confirms that the difference between soft deletion and hard deletion was statistically significant ($p = 3.95 \times 10^{-13}$). The same threat appears in the Wikipedia BLP dataset, which contains real biographical information about living individuals. On 98 BLP records, Vec2Text recovered multiple categories of identifying information from soft-deleted embeddings without any additional fine-tuning or distribution alignment. Recovery rates reached 46.4\% for geographic locations~(GPE), 44.7\% for affiliations~(NORP), 25.5\% for person names (PERSON), and 6.2\% for organizations~(ORG), with a ROUGE-L score of $0.185 (\pm 0.062$), as shown in Table~\ref{tab:scope}. In contrast, the hard-delete baseline produced 0\% recovery across all entity categories. To ensure that the recovered information was coming from deleted records rather than from similar live records, we also tested the 50 most unique deleted records. These were records whose entities did not appear in the remaining live data. Vec2Text achieves 18\% entity recovery, identifying entities from 9 of these 50 records. The deleted vector itself was the only available source of that information, creating a temporal gap risk window, in which an attacker can recover identity-related information from records that were supposed to have been removed from the system.

\begin{table}[ht]
\centering
\small
\begin{tabular}{l l l r}
\toprule
\textbf{Database} & \textbf{Soft-Delete} & \textbf{Persist} & \textbf{ROUGE-L} \\
\midrule
ChromaDB & hnswlib mark & YES & 0.207 \\
FAISS & No remove() & YES & 0.167 \\
Weaviate & SQLite & YES\textsuperscript{\textdagger} & N/A \\
\bottomrule
\end{tabular}

\begin{flushleft}
\caption{Soft-delete persistence across evaluated systems.} 
\label{tab:multidb}
\end{flushleft}
\end{table}

\vspace{-1.5cm}

\begin{table}[h]
\centering
\scalebox{0.9}{
\begin{tabular}{l r r l}
\toprule
\textbf{Quantization} & \textbf{ROUGE-L} & \textbf{\% Baseline} & \textbf{Production Use} \\
\midrule
Float32 (baseline) & 0.207 & 100\% & Research \\
SQ8 (int8) & 0.193 $\pm$ 0.048 & 93\% & Most common \\
PQ32 (product) & 0.162 $\pm$ 0.043 & 78\% & Memory-optimized \\
Binary (1-bit) & 0.123 $\pm$ 0.036 & 59\% & Rare \\
\bottomrule
\end{tabular}}

\caption{Reconstruction Quality vs. Quantization Scheme}
\label{tab:quantization}
\end{table}

\begin{table}[ht]
\centering
\small
\setlength{\tabcolsep}{2pt} 
\begin{tabular}{|p{2cm}|p{0.6cm}|p{1.7cm}|p{1.7cm}|p{1.5cm}|}
\hline
\textbf{Dataset} & \textbf{N} & \textbf{Condition} & \textbf{ROUGE-L} & \textbf{BERTScore} \\
\hline
Wikipedia (orig) & 200 & Soft-Delete & 0.207 ($\pm$0.07) & 0.798 \\
\hline
Wikipedia (orig) & 200 & Hard-Delete & 0.000 ($\pm$0.00) & 0.566 \\
\hline
Wikipedia ($n=500$) & 500 & Soft-Delete & 0.207 ($\pm$0.07) & 0.858 \\
\hline
Wikipedia ($n=500$) & 500 & Hard-Delete & 0.00 ($\pm$0.00) & 0.566 \\
\hline
Clinical (synth) & 1000 & Soft-Delete & 0.085 ($\pm$0.04) & 0.748 \\
\hline
Clinical (synth) & 1000 & Hard-Delete & 0.000 ($\pm$0.00) & 0.566 \\
\hline
Synthea EHR & 1000 & Soft-Delete & 0.290 ($\pm$0.033) & 0.837 \\
\hline
Synthea EHR & 1000 & Hard-Delete & 0.000 ($\pm$0.00) & 0.566 \\
\hline
FAISS HNSW & 500 & Soft-Delete & 0.167 ($\pm$0.03) & N/A \\
\hline
FAISS HNSW & 500 & Hard-Delete & 0.002 ($\pm$0.01) & N/A \\
\hline
\end{tabular}

\caption{Text Reconstruction: Soft-Delete vs. AES.}
\label{tab:reconstruction}
\end{table}

\subsubsection{Clinical Extraction on Structured Data} 
The synthetic clinical dataset shows that soft-deleted records can expose medically relevant identifying information. We analyze 1,000 Faker-generated clinical records, designed to resemble real clinical data. To make them more realistic, the records included noise such as abbreviations, shorthand, and typos; 58.8\% of the records contained at least one of these variations typical of clinical language.  The results show high recovery rates for various types of personal identification information. Gender was the easiest feature to recover, with a recovery rate of 73.4\%. Patient first names were recovered in 65.0\% of cases, and last names in 49.9\%. Diagnosis information was retrieved in 28.0\% of cases, and physicians’ names in 6.2\%. Although some of these have recovery rates, they remain significant because any retrieval of patient identity or medical information still indicates privacy risks under HIPAA and the GDPR.

Table~\ref{tab:pii_recovery} provides the full quantitative breakdown. Cryptographic-erasure/hard deletion baseline eliminated recovery entirely, with all reported PII categories dropping to 0\%. Some fields, such as medical record numbers, drug codes, and dates of birth, could not be recovered by the corrector model. 
A Wilcoxon signed-rank test resulted in \(p = 5.71 \times 10^{-12}\) for the clinical comparison, confirming statistical significance. To test whether recovered identities resulted from random text generation, we ran Vec2Text on 200 L2-normalized, 768-dimensional Gaussian noise vectors and extracted entities with spaCy NER. Although 42\% of the random outputs contained some \texttt{PERSON} entity, none matched the true surname of any individual in the held-out dataset (95\% confidence interval upper bound: 3\%). In contrast, Vec2Text recovered the correct patient surname in 49.9\% of the soft-deleted clinical records, a $17\times$ increase over the random baseline.

\begin{figure}[h]
  \centering
  \includegraphics[width=\columnwidth]{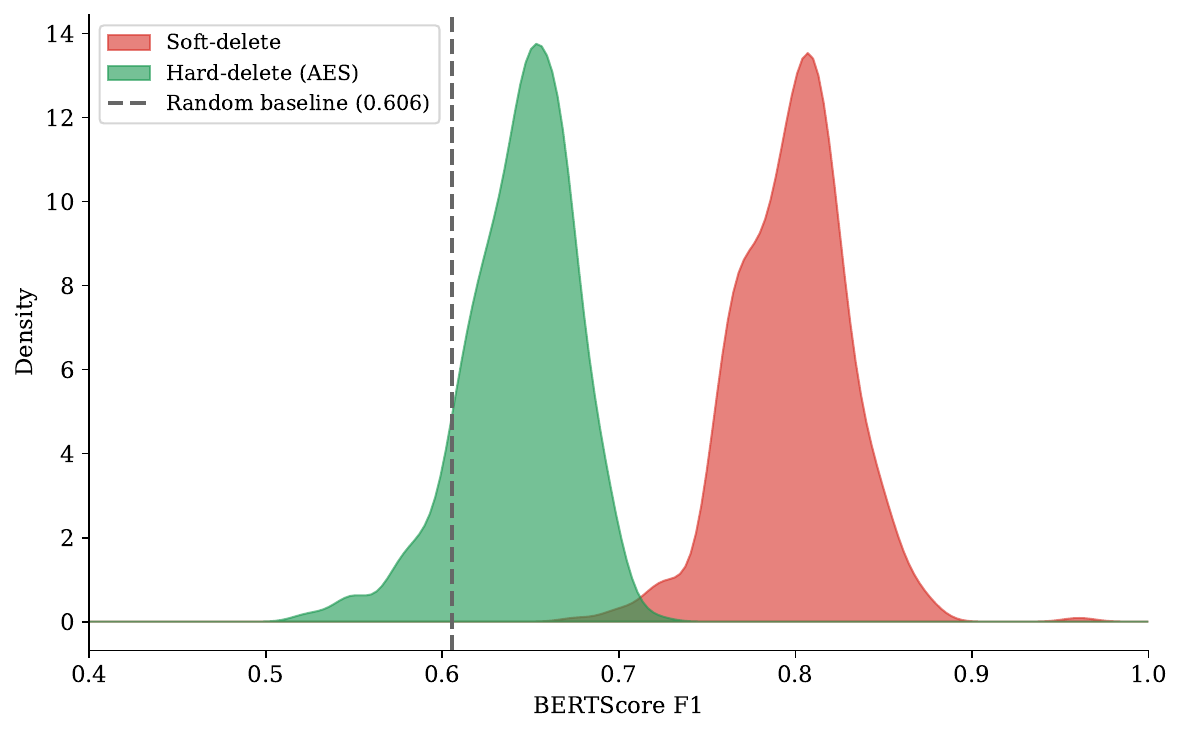}
  \caption{BERTScore F1 distribution for AES-encrypted versus soft-deleted embeddings (Wikipedia, $n=500$). }
  \label{fig:bertscore}
\end{figure}

\begin{table}[h]
\centering
\begin{tabular}{l r r}
\toprule
\textbf{Entity Type} & \textbf{Soft-Delete (\%)} & \textbf{Hard-Delete (\%)} \\
\midrule
Gender & 73.4 (70.6--76.1) & 0 \\
Patient Name & 65.0 (62.0--67.9) & 0 \\
Last Name & 49.9 (46.8--53.0) & 0 \\
Diagnosis & 28.0 (25.2--30.8) & 0 \\
Physician & 6.2 (4.7--7.7) & 0 \\
Hospital & 4.9 (3.6--6.3) & 0 \\
Drug & 3.4 (2.3--4.6) & 0 \\
DOB & 0.0 & 0 \\
MRN & 0.0 & 0 \\
\bottomrule
\end{tabular}
\begin{flushleft}

\caption{Clinical PII Recovery: Soft-Delete vs. AES.}
\label{tab:pii_recovery}
\end{flushleft}
\end{table}

\subsubsection{The Corrector Bottleneck} 
The comparative lower performance of the reconstruction on real clinical notes is due to the limitations of the correction model and not to the storage system itself. Using the zero-shot `gtr-base' corrector on 100 records from the MIMIC-III ICU demo produced very weak reconstruction results, with ROUGE-L scores close to 0.005 for structured records and free-text clinical notes. The Faker clinical dataset approximates an upper bound on recovery for structured clinical text under better corrector alignment. Reconstruction quality improved when the model was adapted to the clinical setting by $26.6\times$. Fine-tuning on 43,324 pairs of MTSamples with oversampled pairs from MIMIC-III resulted in an increase in ROUGE-L from 0.005 to 0.077. Beam-search decoding (\texttt{num\_beams=8}, \texttt{no\_repeat\_ngram\_size=3}) increased ROUGE-L to 0.232 on the retained records (Table V). As a result, 95\% of the retained records exceeded a ROUGE-L score of 0.10, and 85\% exceeded a ROUGE-L score of 0.15. The storage risk remains independent of current corrector’s reconstruction quality. Soft-deleted MIMIC-III vectors still preserve their exact geometric structure inside the HNSW index and remain physically recoverable from disk. Because these indexes preserve structural dependencies, extracting even a single distinct entity, such as a patient surname, is sufficient to create meaningful re-identification risk. We also tested whether the reconstructed text remained semantically similar across different embedding models. The top 50 Wikipedia reconstructions were evaluated using MiniLM, MPNet, and multilingual-MPNet, which are different from the original \texttt{gtr-t5-base} target model. Across all three models, ROUGE-L remained stable at 0.190, while cosine similarity ranged from 0.86 to 0.90 (Table~\ref{tab:surrogate}). This confirms that an attacker does not need access to the victim's exact embedding architecture to recover meaningful information. Recent zero-shot inversion methods, such as Zero2Text~\cite{zero2text2025} and ZSinvert~\cite{zsinvert2025} suggest even stronger cross-model transferability, indicating that these surrogate results represent a conservative baseline for future exploit capabilities.

\begin{table}[ht]
\centering
\small

\begin{tabular}{|p{4.2cm}|p{1.4cm}|p{1.6cm}|}
\hline
\textbf{Corrector} & \textbf{ROUGE-L} & \textbf{Improvement} \\
\hline
Zero-shot \texttt{gtr-base}              & 0.009 & ---        \\
MIMIC-only fine-tuned ($n=80$)           & 0.005 & ---        \\
Mixed fine-tuned (MTSamples + MIMIC 10$\times$) & 0.077 & $8.6\times$ \\
Mixed + beam search (\texttt{num\_beams=8}) & 0.232 & $\mathbf{26.6\times}$ \\
\hline
\end{tabular}
\caption{Domain Adaptation on MIMIC-III Reconstruction}
\label{tab:adaptation}
\end{table}

\begin{table}[h]
\small
\centering
\begin{tabular}{l r r}
\toprule
\textbf{Surrogate Model} & \textbf{Cosine Sim} & \textbf{ROUGE-L} \\
\midrule
MiniLM-L6          & 0.86 & 0.190 \\
MPNet               & 0.88 & 0.190 \\
Multilingual-MPNet  & 0.90 & 0.190 \\
\bottomrule
\end{tabular}
\caption{Cross-architecture Transferability using \texttt{gtr-base}}
\label{tab:surrogate}
\end{table}

\begin{table}[ht]
\centering
\small
\setlength{\tabcolsep}{2pt}
\scalebox{0.94}{
\begin{tabular}{|p{2cm}|p{0.8cm}|p{1.4cm}|p{1.4cm}|p{1.7cm}|}
\hline
\textbf{Dataset} & \textbf{Type} & \textbf{ROUGE-L} & \textbf{BERTScore} & \textbf{Notes} \\
\hline
Wiki ($n=500$)   & Real  & 0.207 & 0.858 & Primary result \\
\hline
Wiki BLP         & Real  & 0.185 & 0.585 & Living persons \\
\hline
Synth Clinical ($n=1{,}000$) & Synth & 0.085 & 0.748 & Dist-aligned upper bound \\
\hline
MIMIC-III Text   & Real  & 0.005 & ---$^\dagger$ & OOD corrector \\
\hline
MIMIC-III EHR    & Real  & 0.005 & ---$^\dagger$ & OOD corrector \\
\hline
MIMIC + Clinical & Real  & 0.232 & ---$^\ddagger$ & Fine-tuned + beam search \\
\hline
OpenAI Embeds    & Real  & 0.020 & ---$^\dagger$ & Arch mismatch \\
\hline
\end{tabular}}
\vspace{0.1cm}
\caption{Attack Effectiveness Across Text Distributions. \\
$^\dagger$BERTScore omitted where ROUGE-L $<0.01$}
\label{tab:scope}
\end{table}

\subsection{Secondary Risks}
\label{sec:secondary}

\subsubsection{Persistence under quantization} 
Vector quantization is commonly used in production systems to reduce storage and memory usage. To evaluate the reconstruction quality, we tested FAISS quantized indexes using binary (1-bit), PQ32 (product quantization), and SQ8 (scalar int8) quantization. The results are summarized in Table~\ref{tab:quantization}. The attack remains partially effective under all tested compression methods. SQ8 and PQ32, two of the most commonly used quantization approaches, preserved between 78\% and 93\% of the original reconstruction quality. These results suggest that PQ-based~\cite{zhuang2024vec2text} compression alone does not eliminate storage-layer recovery. Even after PQ32 compression, ghost vectors preserve significant reconstruction quality, indicating a defense built for live API queries does not address raw index access.

\subsubsection{Cross-Modality Persistence}
\label{sec:multimodal}
In our experiments, any type of data stored in an HNSW index as float32 vectors remain recoverable after soft-deletion. We use CLIP ViT-H/14 embeddings ($d=1{,}024$) along with ChromaDB to test biometric and medical image data protected under Article 9 of the GDPR. The same extraction process was used as in the text-based attack. The downstream evaluation uses a $k$-NN classifier ($k=5$, cosine similarity), trained on live embeddings and tested on ghost embeddings~(Table~\ref{tab:multimodal}).

\textbf{Tier~1: Facial identity recovery -}
For biometric data, ghost embeddings show 99.17\% top-1 identity recovery on the Labeled Faces in the Wild (LFW) dataset, consisting of 158 individuals and 4,324 images. After soft-deletion, the cosine similarity of ghost embeddings remains at 1.000~(Figure~\ref{fig:multimodal_umap}a). The attack outperforms the 0.63\% random baseline by 156.7$\times$ (Wilcoxon $p < 10^{-12}$}). Since GDPR Article 9 requires explicit consent for processing biometric data used for identification, persistence of deleted embeddings raises significant privacy and compliance concerns.

\textbf{Tier~2: Histopathology class preservation -}
Using PathMNIST histopathology images, tissue class information remained recoverable at 100\% accuracy, with ghost cosine similarity again remaining at 1.000~(Figure~\ref{fig:multimodal_umap}b). The $2.0\times$ lift represents the mathematical ceiling for balanced binary classification. Additional tests with domain-specialized models (torch-xrayvision DenseNet-121 and BiomedCLIP~\cite{damm2025overview}) produced downstream classifiers that failed to separate classes effectively because inter-class cosine similarity remained extremely high ($\ge 0.98$). This negative result is still informative because it shows that ghost vector persistence is independent of the strength of a  downstream classifier. Soft-deleted chest radiographs remain physically present. The storage-layer threat may extend to multimodal systems built on HNSW indexes. The same storage-layer threat affects text, biometric, and medical image embeddings.

\begin{figure*}[t]
\centering
\includegraphics[width=0.8\linewidth]{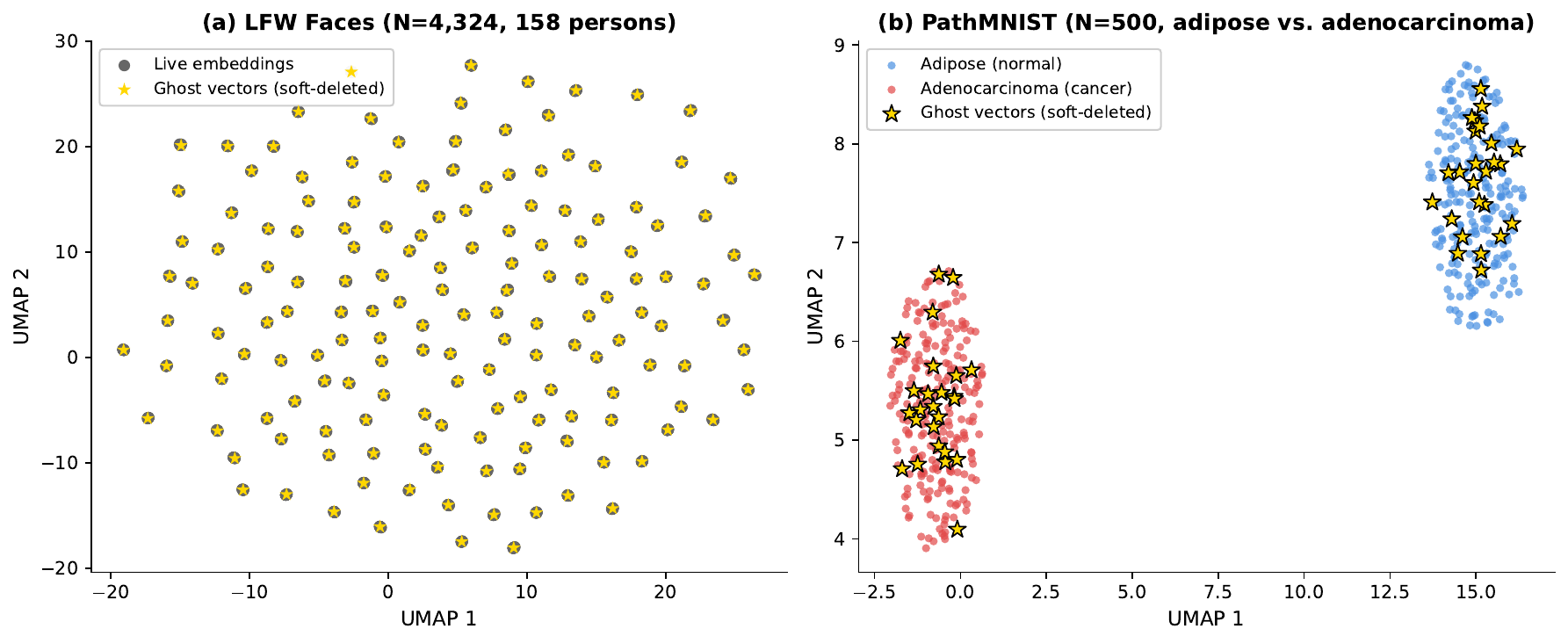}
\caption{
UMAP projections of CLIP ViT-H/14 embeddings. \textbf{(a)} LFW faces. \textbf{(b)} PathMNIST patches. Gold stars mark ghost vectors of soft-deleted images, perfectly overlapping their original clusters.}
\label{fig:multimodal_umap}
\end{figure*}

\begin{table*}[t]
\centering
\small
\begin{tabular}{@{}llrrrrrrr@{}}
\toprule
\textbf{Tier} &
\textbf{Dataset} &
\textbf{Embedding Model} &
\textbf{$d$} &
\textbf{$N$ img} &
\textbf{$N$ deleted} &
\textbf{Ghost Sim} &
\textbf{Recovery} &
\textbf{Lift} \\
\midrule
1 & LFW Faces~\cite{LFWTech} & CLIP ViT-H/14 & 1{,}024 & 4{,}324 & 158 & \textbf{1.000} & \textbf{99.17\%} & \textbf{156.7$\times$} \\
2 & PathMNIST~\cite{medmnistv2} & CLIP ViT-H/14 & 1{,}024 & 1000 & 50 & \textbf{1.000} & \textbf{100.0\%} & \textbf{2.0$\times\dagger$} \\
\bottomrule
\end{tabular}
\smallskip

\noindent
\footnotesize{$\dagger$ 2.0$\times$ represents the mathematical ceiling for balanced binary classification.}
\caption{\textbf{Multimodal Ghost Vector Recovery Results.}}
\label{tab:multimodal}
\end{table*}

\subsection{Structural Persistence and Side Channels}
\label{sec:ghost_traversal}
Soft-delete operates at the metadata layer, instead of the graph layer. In the HNSW graph, deleted nodes remain linked to their neighbors. The greedy walk still visits ghost nodes during search---they are filtered from output but not from traversal. This structural persistence creates a secondary risk, that can appear as a query-level side channel.

\subsubsection{Ghost Traversal Cost} 
Using $N=1{,}000$ Wikipedia embeddings (\texttt{gtr-t5-base}, $d=768$, $M=32$), we construct an FAISS \texttt{IndexHNSWFlat} instance and isolate 50 target vectors for deletion. We then generate 100 probe queries per target using a tight Gaussian distribution ($\sigma=0.01$). Under soft deletion, FAISS reported 672.7 distance computations per query, which was nearly identical to the pre-deletion baseline. A Wilcoxon test also showed no statistically meaningful difference ($p=0.97$). Query latency stayed the same as well, at $34.8\,\mu$s before and after soft deletion. On the other hand, as seen in Table~10, true deletion reduced latency to $33.3\,\mu$s, a 4.3\% reduction ($p=0.025$). In disk-based indices such as DiskANN~\cite{gaussdb2025}, this cost could be higher because each wasted ghost-node traversal may require additional storage access. Overall, the evaluation shows the HNSW greedy walk processes soft-deleted nodes with the same computational cost as live nodes.

\subsubsection{Result-Set Drift}
Soft deletion also impacts search results. As shown in Table~\ref{tab:ghost_traversal}, 95\% of top-$K$ neighbor sets differ between the soft-delete and true-deletion conditions. Ghost nodes can still function as structural hubs that guide the greedy walk through the embedding space. In this sense, the deleted record still `votes', or influences the search process. Soft deletion can change results for live queries. Records that users expected to be removed can influence other users' search results due to deleted nodes remaining structurally present.

\subsubsection{API-Layer Side Channel}
The difference in latency between soft deletion and true deletion also creates a timing side channel. Since true deletion reduced latency by 4.3\%, an attacker with API query access can compare latency against a pre-deletion baseline and attempt to infer whether a node had been soft-deleted or actually removed. This extends the ghost vector threat beyond raw storage access and into the query API layer.

\subsubsection{Defense Requirements}
As detailed in the following section, proposed Epoch key rotation prevents reconstruction by encrypting the ghost vector's contents and destroying the key. However, this only addresses the content of the vector. The ghost node can still remain in the HNSW graph. A more complete defense requires two steps: 1) content erasure through epoch key rotation and 2) graph repair through neighbor re-linking. In graph repair, the deleted node's predecessors would be connected directly to its successors.

\subsubsection{Adversary Cost}\label{sec:cost}
A single inference step during evaluation achieves 98.4\% of plateau reconstruction quality. On $n=50$ Wikipedia records, we perform a step-count sweep at step counts 1, 5, 10, 20, and 50. ROUGE-L = 0.204 is recovered in a single step at 30.7 ms/record. At \texttt{n\_steps=10} (0.2073), ROUGE-L reaches a plateau near 0.207, and additional steps have no effect (Figure~\ref{fig:stepcount}). On a single GPU, a single-pass inference can invert one million soft-deleted records~(extrapolating from 30.7 ms/record) in less than 9 hours ($10^6 \times 30.7\,\text{ms} \approx 8.5\,\text{hours}$). This verifies the attacker requires neither specialized hardware nor iterative refinement.

\begin{figure}[h]
\centering
\includegraphics[width=\columnwidth]{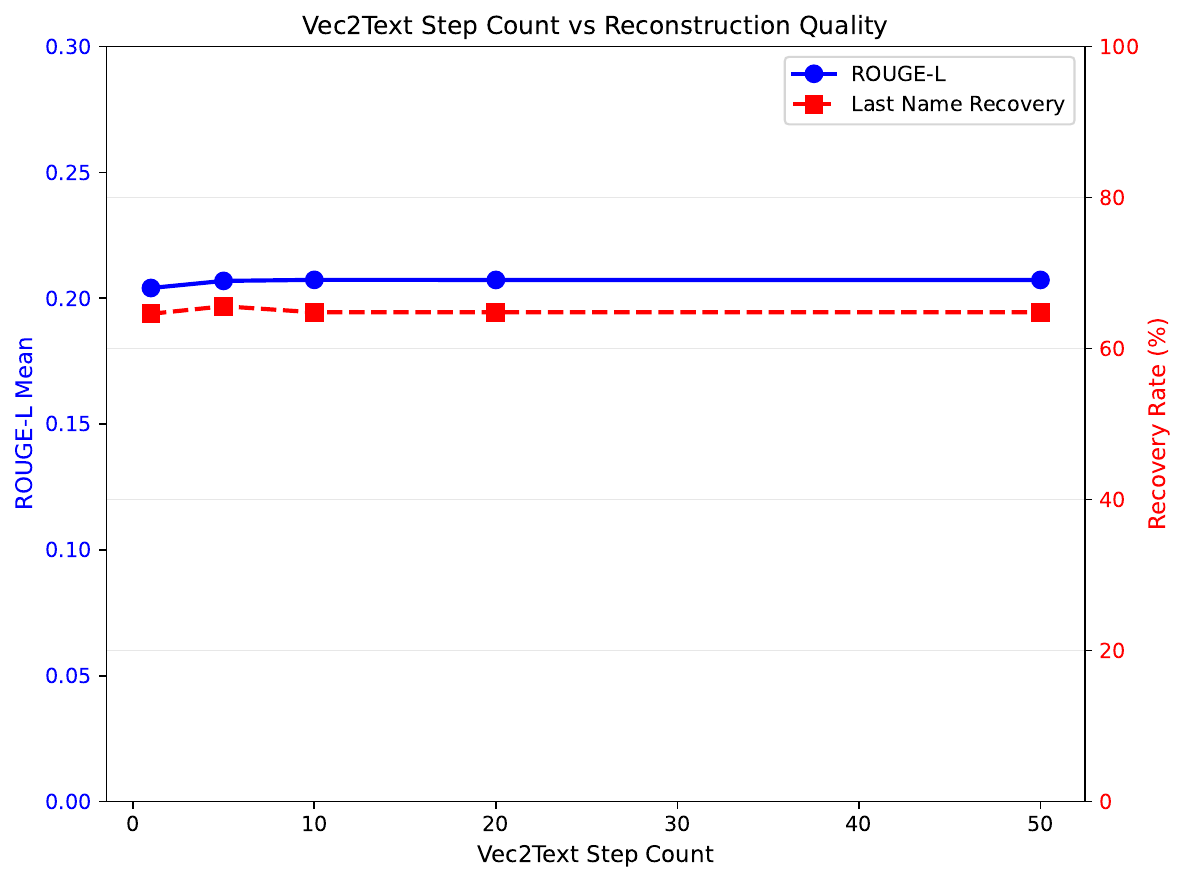}
\caption{Vec2Text step count vs reconstruction quality. A single inference step achieves 98.4\% of plateau quality at 30.7ms/record.}
\label{fig:stepcount}
\end{figure}

\begin{table}[ht]
\centering
\small
\setlength{\tabcolsep}{3pt}
\begin{tabular}{|p{2.4cm}|p{1.6cm}|p{1.6cm}|p{1.6cm}|}
\hline
\textbf{Metric} & \textbf{A (before)} & \textbf{B (soft-del)} & \textbf{C (true del)} \\
\hline
ndis per query & 672.7 & 672.7 & 669.9 \\
\hline
Latency ($\mu$s) & 34.8 & 34.8 & 33.3 \\
\hline
Wilcoxon $p$ (vs A) & --- & 0.97 & 0.025 \\
\hline
\multicolumn{4}{|p{7.6cm}|}{Result set drift B vs. C: \textbf{95.0\%} of neighbors differ} \\
\hline
\end{tabular}

\caption{Ghost Traversal: Three-Condition Comparison.}
\label{tab:ghost_traversal}
\end{table}

\section{Epoch Key Rotation Defense}
\label{sec:defense}

\subsection{Formal Specification}
\label{sec:defense_spec}
\begin{definition}[Primitive] 
Let $\mathcal{E} = (\mathsf{KeyGen}, \mathsf{Enc}, \mathsf{Dec})$ be a symmetric encryption scheme. An epoch key rotation primitive for a vector database $\mathcal{D}$ consists of a tuple $(\mathsf{EpochKeyGen}, \mathsf{Rotate}, \mathsf{Prove})$ where:
\begin{itemize}
  \item $\mathsf{EpochKeyGen}(user\_id, epoch\_id) \to (k_e, metadata)$: Generates a fresh 256-bit epoch key, stores only non-secret (user\_id, epoch\_id, creation\_timestamp) in SQLite and returns a key handle (not the key itself).
  \item $\mathsf{Rotate}(user\_id, \mathcal{D}_u) \to (\mathcal{D}_u', \pi)$: Loads the current epoch key for user $u$. AES-256-CTR encrypts all vectors and generates a new epoch key for subsequent writes and returns the encrypted embeddings together with proof $\pi$.
  \item $\mathsf{Prove}(rotation\_event) \to \pi$: Generates SHA-256 event hash of $(user\_id + old\_epoch\_id + new\_epoch\_id + timestamp)$, signs with ECDSA-SHA256 using the controller's private key (SECP256R1). It finally returns a JSON proof object verifiable with registered public key.
\end{itemize}
\end{definition}

Epoch granularity is per-user (each $user\_id$ has an independent $epoch\_id$ counter, incremented on deletion request). The epoch key store must be present in a separate security domain from the HNSW binary (e.g., KMS, HSM). Standard cloud architectures meet this by separate IAM-controlled services for storage and key management. Any vector encrypted under a discarded $k_e$ is computationally indistinguishable from random noise. The ciphertext reinterpreted as float32 does not preserve semantic structure for an inversion model to exploit.

\subsection{Failure of Prior Defenses}
Prior defenses against inversion attacks mainly function at query-time and are ineffective at the storage layer. We empirically show that query-side embedding transformations~(a representative query-side defense proposed against embedding inversion) fail at the storage layer. To guard against embedding inversion, current works~\cite{zhang2024embedding} propose using a random orthogonal transformation $T$. However, any adversary with access to the server-side configuration can trivially bypass this mitigation by extracting $T$ and applying $T^\top$.

To confirm, we generate a random orthogonal $768 \times 768$ matrix $T$ using QR decomposition of a Gaussian matrix. We then transform 500 soft-deleted embeddings and apply Vec2Text to the converted vectors. Before reversing the transformation with $T^\top$, ROUGE-L score was at 0.073. Afterward, the score rose to 0.215, nearly matching the baseline inversion attack score of 0.207 (Table~\ref{tab:transformation}). This shows that embedding transformation defenses fail against storage-layer adversaries since the transformation key $T$ must remain a persistent secret that can be extracted directly from server configuration files. Epoch key rotation prevents this architectural vulnerability since the AES key is cryptographically destroyed on rotation, leaving no recoverable secret recoverable.

\begin{table}[h]
\centering
\small
\begin{tabular}{l r r}
\toprule
\textbf{Condition} & \textbf{ROUGE-L} & \textbf{vs. Baseline} \\
\midrule
No transformation (soft-delete) & 0.207 & --- \\
After applying $T$ (without $T^\top$) & 0.073 & $-$65\% \\
After applying $T$ + inverting $T^\top$ & 0.215 & $+$4\% \\
\midrule
Epoch rotation (AES, key destroyed) & 0.000 & $-$100\% \\
\bottomrule
\end{tabular}
\caption{Orthogonal transformation defense}
\label{tab:transformation}
\end{table}

\begin{table}[h]
\centering
\small
\begin{tabular}{|p{1.5cm}|p{1.3cm}|p{1.3cm}|p{0.8cm}|p{0.8cm}|p{0.7cm}|}
\hline
\textbf{Approach} & \textbf{PII Rec.} & \textbf{ROUGE-L} & \textbf{R@10} & \textbf{Time} & \textbf{Proof} \\
\hline
Soft-delete & High (49.9/ 73.4\%) & 0.2073 & 100\% & $\sim$0ms & No \\
\hline
Full rebuild & High & 0.2073 & 100\% & 2.25s (N=500) & No \\
\hline
Per-rec. AES & 0\% & 0.0000 & 90.4\% & 8ms & No \\
\hline
Epoch rot. & 0\% & 0.0000 & 90.4\% & 2.5ms & Yes \\
\hline
\end{tabular}
\caption{Defense Comparison. Speedup relative to full rebuild, $N$=100,000.}
\label{tab:defense}
\end{table}

\subsection{Evaluation}
\label{sec:defense_eval}
Epoch key rotation shows approximately linear scaling. We measure total runtime at $N=100$, $N=500$, and $N=5{,}000$ records, obtaining 0.51 ms at $N=100$, 2.48 ms at $N=500$, and 24.93 ms at $N=5{,}000$ respectively.  These measurements correspond to an amortized AES throughput of approximately 0.005 ms per record. Extrapolating to $N$=100,000 gives about 550 ms for vector encryption. This confirms linear scaling of AES-256-CTR throughput. ECDSA proof generation adds $\sim$3 ms, bringing the full protocol to $\sim$553 ms. Structural graph repair (neighbor re-linking) is discussed in Section~\ref{sec:limitations}.

In the real world, epoch-encrypted vectors face minimal retrieval degradation when they remain in the graph. Since retrieval degradation is proportional to the ratio of encrypted to live vectors, the degradation remains within tolerance for most RAG deployments. Systems needing higher retrieval precision should schedule frequent physical cleanup of encrypted vectors. The cryptographic proof $\pi$ is generated at rotation time and remains independently verifiable even if encrypted vectors are physically removed later. We empirically compare four deletion strategies below on the same 500 Wikipedia soft-deleted vectors. The results are summarized in Table~\ref{tab:defense}. \\
\textbf{i) Soft-delete}: Time: $\sim$0 ms, ROUGE-L=0.207, without proof~(baseline). \\
\textbf{ii) Full index rebuild}: Time: 2{,}250ms, ROUGE-L=0.207 (same as soft-delete), without proof. Index rebuild is not retroactive. An adversary who copies the binary before rebuild achieves identical reconstruction quality. \\
\textbf{iii) Per-record AES}: Time: 8 ms, ROUGE-L=0.000, without proof. Effective but 3.2$\times$ slower than epoch rotation and produces no verifiable evidence. \\
\textbf{iv) Epoch key rotation}: Time: 2.5 ms at N = 500, ROUGE-L=0.000, cryptographic proof~$\pi$~($\sim$300$\times$ faster than full index rebuild and 53$\times$ faster than \textbf{an insert-only} rebuild at $N=100{,}000$). End-to-end runtime is about 553 ms including proof generation. This approach provides verifiable erasure.\\

\begin{table}[h]
\centering
\small
\begin{tabular}{|p{0.8cm}|p{0.7cm}|p{1cm}|p{1.7cm}|p{0.9cm}|p{1.1cm}|}
\hline
\textbf{$\epsilon$} & \textbf{$\sigma$} & \textbf{ROUGE-L} & \textbf{95\% CI} & \textbf{R@10} & \textbf{Cos. Sim} \\
\hline
100   & 0.048 & 0.0830 & [0.082, 0.084] & 55.8\% & 0.950 \\
\hline
50    & 0.097 & 0.0760 & [0.075, 0.077] & 26.6\% & 0.850 \\
\hline
20    & 0.242 & 0.0747 & [0.074, 0.076] & 4.8\%  & 0.550 \\
\hline
10    & 0.484 & 0.0756 & [0.075, 0.077] & 0.4\%  & 0.079 \\
\hline
5     & 0.969 & 0.0749 & [0.074, 0.076] & 0.1\%  & 0.025 \\
\hline
2     & 2.419 & 0.0734 & [0.072, 0.074] & 0.0\%  & 0.008 \\
\hline
1     & 4.839 & 0.0748 & [0.074, 0.076] & 0.0\%  & 0.005 \\
\hline
0.5   & 9.678 & 0.0749 & [0.074, 0.076] & 0.0\%  & 0.004 \\
\hline
0.1   & 48.39 & 0.0730 & [0.072, 0.074] & 0.0\%  & 0.001 \\
\hline
0.05  & 96.78 & 0.0730 & [0.072, 0.074] & 0.0\%  & 0.000 \\
\hline
0.01  & 483.9 & 0.0730 & [0.072, 0.074] & 0.0\%  & 0.000 \\
\hline
\textit{Epoch rot.} & --- & \textbf{0.000} & {[0, 0]} & \textbf{0.0\%} & N/A \\
\hline
\end{tabular}
\caption{Differential Privacy vs.\ Vec2Text Recovery}

\label{tab:dp}
\end{table}

\subsection{Why Differential Privacy Fails}
\label{sec:dp_failure}
Our findings show Gaussian DP noise creates a poor privacy-utility tradeoff for deleted embeddings. Starting from 500 soft-deleted Wikipedia embeddings, we apply 11 privacy budgets: $\epsilon \in \{100, 50, ... 0.05, 0.01\}$ (Table~\ref{tab:dp}), with $\delta=10^{-5}$. The Gaussian mechanism uses $\sigma = \sqrt{2 \ln(1.25/\delta)} \cdot \Delta f / \epsilon$ where sensitivity $\Delta f = 1$ for L2-normalized embeddings. We apply Vec2Text on each noise-perturbed embedding ($n=10$ steps) and calculate ROUGE-L recovery and cosine similarity with the original embedding.  Cosine similarity falls to 0.079 at $\epsilon=10$ (weak privacy), indicating that the embedding is semantically destroyed, while ROUGE-L stays at 0.082. At $\epsilon=0.1$, ROUGE-L remains around 0.073 while cosine similarity drops to 0.001~(confirming structural geometry is completely lost). The mechanism of failure is demonstrated by the decoupling between ROUGE-L and cosine similarity: the embedding space is collapsed by DP noise, but Vec2Text smoothly degrades into fluent generic text that still has overlap with the target. ROUGE-L remains around 0.073–0.083 across the tested \(\epsilon\) values. Thus, DP suppresses semantic recovery, but does not eliminate it. Any $\epsilon$ strong enough to prevent reconstruction also makes the index unusable for its primary retrieval purpose.

\textbf{The Utility-Privacy Trade-Off:} 
The utility-privacy trade-off fails for deleted data. ROUGE-L scores below 0.05 would require noise levels that destroy retrieval utility. Even $\sigma=0.48$ ($\epsilon=10$) adds a noise magnitude of $\sqrt{768} \cdot 0.48 \approx 13.4$ in 768-dimensional space, rotating the embedding significantly (cosine similarity $\approx 0.08$). Preserving retrieval utility ($\text{recall@10} \ge 0.8$) requires cosine similarity $> 0.85$ and $\sigma < 0.02$, corresponding to $\epsilon > 200$---a budget that offers no privacy protection. For protecting deleted records, DP introduces noise to live records. When the record is live, the noise must be small enough to preserve retrieval. This guarantees the noise is too small to protect the record once deleted. Epoch rotation cryptographically renders deleted data unrecoverable under the stated assumptions rather than attempting to obscure it. A deleted record becomes completely irrecoverable by construction. Epoch key rotation results in ROUGE-L = 0.000 at minimal utility cost for remaining records~(Figure~\ref{fig:dp_tradeoff}).

\begin{figure}[h]
\centering
\includegraphics[width=\columnwidth]{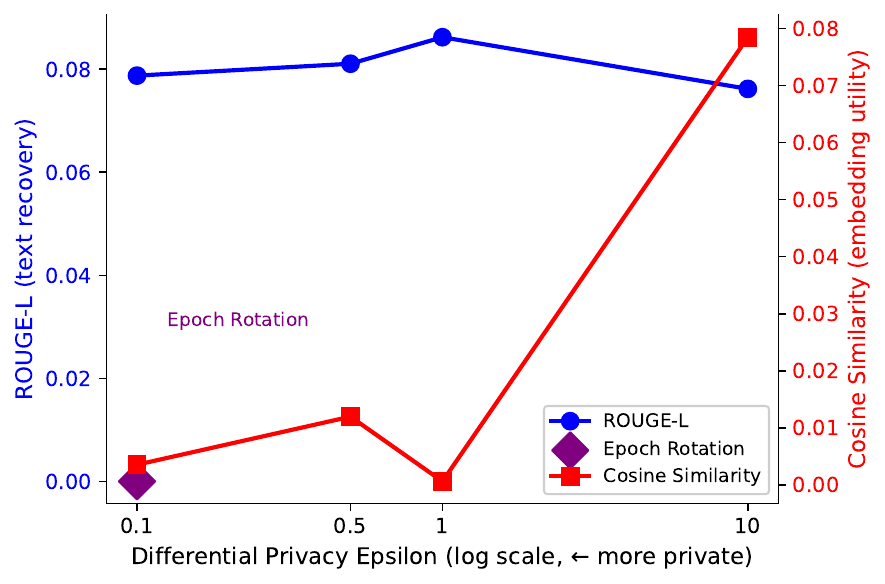}
\caption{Differential Privacy vs.\ Retrieval Utility Tradeoff. x-axis shows $\epsilon$ (log scale; privacy increases right-to-left). Left y-axis: ROUGE-L (text recovery). Right y-axis: cosine similarity (embedding utility)}
\label{fig:dp_tradeoff}
\end{figure}

\vspace{-0.2 cm}

\section{Discussion}
\label{sec:discussion}

\textbf{GDPR Compliance and Legal Implications:}

This paper shows data persistence at the storage layer raises serious GDPR compliance concerns. Article 17 requires verifiable erasure, not just functional deletion. An API returning no records can mean either deletion or suppression. Soft-delete, therefore, provides no real proof of deletion. Three explicit regulatory requirements include Art.~5(1)(e) storage minimization (data cannot be kept past its intended use), Art.~17(1) erasure right (deleted data should no longer remain recoverable), and Art.~5(2) accountability (proof $\pi$ provides auditable evidence). As per EDPB guidelines~(05/2019), erasure must be `verifiable and irreversible'~\cite{edpb2019erasure}. Thus, soft deletion raises compliance concerns when deleted vectors remain recoverable from storage. Epoch key rotation supports this requirement by destroying the cryptographic keys needed to recover the data and generating ECDSA proofs that can be audited.

\textbf{Scope, Controller/Processor, and HIPAA:}
Article 17 becomes applicable when a RAG corpus contains personal data, including consumer-facing applications and SaaS deployments~(Article 28). The compliance gap is more severe for multimodal RAG systems handling Article 9 special-category data, such as facial biometrics and histopathology images. Within U.S. healthcare context, NIST SP 800-88 categorizes cryptographic erasure as a PHI sanitization technique.  The storage-layer persistence finding remains completely independent of transient reconstruction quality. This structural survival serves as the main compliance argument for free-form clinical notes, where the immediate re-identification threat depends on corrector availability. Under EDPB Guidelines 05/2019, deletion must extend to all storage layers, including backup snapshots~\cite{edpb2019erasure}. Effective compliance requires a breach-resistant erasure mechanism. By producing $\pi$ in response to an Article 17 audit, a data controller satisfies the demonstrability requirement regardless of attack success, providing a stronger step toward cryptographic proof of key destruction.

\textbf{Compliance debt:}
We introduce the concept of compliance debt to describe the accumulated compliance burden: each deletion request can leave one more recoverable record at the storage layer. Epoch key rotation addresses this burden by rendering deleted vectors unrecoverable under the stated assumptions and by generating a signed proof~$\pi$.

\section{Limitations and Future Work}
\label{sec:limitations}
\textbf{Corrector Generalization and Scale: }
Two main factors limit the current work. First, the clinical evaluation uses synthetic records, and validation on raw MIMIC-III clinical notes showed a ROUGE-L score of 0.009. Because the \texttt{gtr-t5-base} corrector was trained on Wikipedia-style text, it does not generalize well to clinical terminology. We argue that the bottleneck is due to the corrector rather than the storage. On disk, the soft-deleted MIMIC-III vectors remain physically retrievable. Near-term extensions include clinical corrector training. Initial fine-tuning on $N = 43{,}324$ MTSamples pairs improves MIMIC-III ROUGE-L from 0.009 to 0.052 (5.9$\times$). Additional domain-specific training on MIMIC-III or i2b2 may improve recovery further. Second, our experiments span $10^3$ to $10^5$ records, which is smaller than production deployments that typically handle between $10^7$ and $10^10$ records. Computational constraints prevented larger trials, although the geometric properties of HNSW suggest that the attack should remain feasible at larger scales.  Future research should evaluate additional cloud-hosted vector databases, including Pinecone and Qdrant, to determine whether comparable persistence behavior appears under their storage designs.

Multimodal ghost vectors also raise intellectual-property risks. Source code embeddings (e.g., CodeBERT, GraphCodeBERT) may introduce analogous ghost vector IP theft. Soft-deleted proprietary function logic remains as semantic ghost embeddings in any HNSW-backed code search or copilot system, extending the privacy threat beyond personal data (GDPR) to IP protection.  Some Multimodal RAG systems in production may use a shared HNSW index to combine text, image, and code embeddings. In a single storage breach, a ghost vector adversary could also recover deleted text PII, biometric identity, and proprietary code, further compounding the threat. Epoch key rotation extends naturally to this case since AES-256-CTR is modality-blind.

\textbf{Hardware Isolation and Computational Storage:}
Our evaluation of epoch key rotation is implemented as a software prototype. This creates a trust dependency on the host OS. In the event of a breach, a root-level adversary with access to the host CPU can potentially extract active AES-256 epoch keys from RAM before rotation. Mitigating this dependency requires shifting the cryptographic boundary to hardware via Computational Storage Drives (CSDs) with embedded ARM cores or FPGAs, such as the Samsung SmartSSD~\cite{chakraborttii2021improving}. By offloading the index to the drive, the host OS never sees the raw float32 vectors or the epoch keys, because HNSW traversal and AES decryption are executed directly on the drive controller. A hardware-assisted design thus could move key invalidation closer to the storage device. This creates a path towards stronger isolation for vector-database deployments and ensures cryptographic erasure even when the host kernel is fully compromised. One limitation that no corrector advance can address is forward secrecy. Similar to TLS sessions recorded before a key rotation, vectors extracted before epoch rotation remain invertible. This represents a fundamental design limitation, rather than a flaw in the current implementation.

\section{Related Work}
\label{sec:related}
Prior research showed that dense embeddings can leak semantic information~\cite{song2020information}. Follow-up works developed both  zero-shot~\cite{zero2text2025, zsinvert2025} and iterative~\cite{morris2023vec2text, zhuang2024vec2text} inversion methods. Zhuang et al.~\cite{zhuang2024vec2text} presented a detailed study of inversion risks and query-layer defenses, including embedding transformations~\cite{morris2024vec2textthreat}. Most earlier studies assume an adversary at the API layer targeting \emph{live} embeddings. When this threat model is extended to the storage layer, the semantic vulnerability can remain even after data is logically removed, and query-side transformations remain reversible via $T^\top$.

To avoid the severe I/O penalties of graph restructuring~\cite{wang2025hnswdelete}, vector databases such as ChromaDB~\cite{chromadb}, FAISS~\cite{johnson2019faiss, douze2024faiss}, Weaviate~\cite{weaviate} or Pinecone~\cite{pinecone} use metadata-driven soft-deletion. In some systems,  physical removal is deferred until an asynchronous batch cleanup. This can create a window during which ghost vectors remain reconstructible on disk under the paper’s threat model. Multiple modern system-level defenses address related parts of this problem. MemTrust~\cite{memtrust2026} has been proposed to protect AI memory systems, which is related to our work but addresses a different systems boundary than post-deletion erasure. Alqithami~\cite{alqithami2025forgetful} explores privacy-aware memory architectures that are adjacent to graph-level deletion concerns, but not identical to the storage-layer cryptographic erasure problem studied here. Our epoch key rotation protocol can complement such architectures as a storage-layer cryptographic mechanism for deleted embeddings. Unlike asynchronous batch cleanup alone, the protocol generates a signed proof object ($\pi$) at deletion time.

Machine unlearning studies the impact of removing training influence on neural network model weights~\cite{cao2015machineunlearning, ginart2019making, lizhao2025model}. These approaches need high computation and operate at the model layer, but do not address the persistence of ghost vectors in an HNSW index. Similarly, membership inference attacks~\cite{shokri2017membership} target a different privacy threat surface based on model outputs or query behavior that a breach-level storage adversary cannot access. Prior works also explored secure data expiration. Vanish~\cite{geambasu2009vanish} and Perlman's Ephemerizer~\cite{perlman2005ephemerizer} used time-bounded cryptographic keys to enable data expiration. By adapting this idea to contemporary vector databases, epoch key rotation bridges the gap between the physical realities of storage~\cite{gutmann1996secure} and secure deletion compliance requirements~\cite{politou2018forgetting}.  Finally, solutions such as indirect prompt injection defenses~\cite{greshake2023indirect}, clinical de-identification~\cite{stubbs2015automated, dernoncourt2017identification}, and Cloaked AI~\cite{ironcore2023cloakedai} protect active RAG pipelines, but they are orthogonal to our threat model. They do not address storage-layer persistence or cryptographic proof of deletion under the threat model studied here.

\vspace{-0.2 cm}

\section{Conclusion}
\label{sec:conclusion}
This paper reviews how soft-deletion in modern vector databases can create an illusion of compliance. Logical deletion in RAG-based applications~(such as ChromaDB, FAISS, and Weaviate) leaves high-dimensional embeddings physically intact within the HNSW index. The persistence lies entirely at the storage layer, not prevented by API-level access control under the threat model. Basic inversion models successfully recovered sensitive data from the raw storage directory, including patient surnames, biometric facial identities, and clinical attributes from `deleted' records in the evaluated datasets. The threats extends across modalities and bypasses API-level access since float32 persistence is a basic property of the system. The evaluated alternatives do not fully address this problem. If an adversary copies the binary before a rebuild runs, full index rebuild later provide no retroactive protection. Differential privacy showed a poor privacy-utility: noise levels high enough to suppress recovery also destroys retrieval utility. Epoch key rotation closes this gap at the storage layer. Under the paper’s assumptions, epoch key rotation completes in 2.5 ms for 500 deleted vectors and reduces observed recovery to 0\% while generating a signed proof object \(\pi\) at rotation time. When data controllers rely on soft-delete, deleted data may remain physically recoverable even after it disappears from the API. Epoch key rotation addresses this problem at the storage layer by converting silent suppression into erasure that can be verified cryptographically. Whether the data is text, images, or special-category medical data under GDPR Article~9, epoch key rotation replaces assumed compliance with verifiable proof.

\section{Acknowledgment}
The authors take full responsibility for the technical content of this paper. AI-assisted writing tools were used only for limited language editing in accordance with conference policy.

\bibliographystyle{ieeetr}
\bibliography{references}

\end{document}